\documentclass{llncs}
\usepackage{makeidx}  
\usepackage{etex}
\usepackage{amsmath}
\usepackage{amssymb}
\usepackage[colorlinks = true]{hyperref}
\usepackage[normalem]{ulem}
\usepackage[margin = 2.5cm]{geometry}
\usepackage[english]{babel}
\usepackage[utf8x]{inputenc}
\usepackage{amsmath}
\usepackage{tikz}
\usepackage{graphicx}
\usepackage[usenames,dvipsnames]{pstricks}
\usepackage{epsfig}
\usepackage{pst-grad} 
\usepackage{pst-plot} 
\usepackage{epstopdf}
\usepackage{caption}
\usepackage{subcaption}
\usepackage{hyperref}
\usepackage[ruled,linesnumbered]{algorithm2e}
\usepackage[noend]{algpseudocode}
\makeatletter
\def\BState{\State\hskip-\ALG@thistlm}
\makeatother
\newcommand\Fontvi{\fontsize{7.5}{8.5}\selectfont}

\begin{document}
%
%
%

%
%
\title{Linear time algorithm to check the singularity of block graphs
\thanks{This work is supported by the Joint NSFC-ISF Research Program (jointly funded by
the National Natural Science Foundation of China and the Israel Science Foundation (Nos.
11561141001, 2219/15).}}

\author{Ranveer Singh \inst{1} \and Naomi Shaked-Monderer\inst{2} \and Abraham Berman\inst{3}}

\authorrunning{Ranveer Singh et al.} 

\tocauthor{Ranveer Singh, Naomi Shaked-Monderer, Avi Berman}

\institute{Technion-Israel Institute of Technology, Haifa 32000, Israel \\
\email{singh@technion.ac.il},
\and The Max Stern Yezreel Valley College, Yezreel Valley 19300, Israel\\
\email{nomi@technion.ac.il},
\and
Technion-Israel Institute of Technology, Haifa 32000, Israel\\ \email{berman@technion.ac.il}
}

\maketitle              

\begin{abstract}
A block graph is a graph in which every block is a complete graph. Let $G$ be a block graph and let $A(G)$ be its (0,1)-adjacency matrix. Graph $G$ is called nonsingular (singular) if $A(G)$ is nonsingular (singular).  Characterizing nonsingular block graphs is an interesting open problem proposed by Bapat and Roy in 2013. In this article, we give a linear time algorithm to check whether a given block graph is singular or not. 
\end{abstract}
  
\section{Introduction}  
Let $G=(V(G),E(G))$ be a undirected graph with vertex set $V(G)$ and edge set $E(G)$. Let the cardinality $|V(G)|$ (also called \emph{the order of $G$}) be equal to $n$. The adjacency matrix $A(G)=(a_{ij})$ of $G$ is the square matrix of order $n$ defined by $$a_{ij}=\begin{cases}
1 & \mbox{if the vertices $i,j$ are connected by an edge},\\
0 & \mbox{if $i=j$ or $i,j$ are not connected by an edge},
\end{cases} $$ where $1\leq i,j\leq n.$ A graph $G$ is called \emph{nonsingular (singular)} if $A(G)$ is \emph{nonsingular (singular)}, that is, the determinant of $A(G)$ is nonzero (zero). The \emph{rank} of $G$, denoted by $r(G),$ is the rank of the adjacency matrix $A(G)$. If $G$ has full rank, that is, $r(G)=n$, then $G$ is nonsingular, otherwise, it is singular. The \emph{nullity} of $G$, denoted by $\eta(G),$ is equal to the number of zero eigenvalues of $A(G)$. By the rank-nullity theorem $r(G)+\eta(G)=n.$  Thus a zero nullity of a graph $G$ implies that it is nonsingular while a positive nullity implies that it is singular. A \emph{cut-vertex} of $G$ is a vertex whose removal results in an increase in the number
of connected components. Let $G_1=(V_1, E_1)$ and $G_2=(V_2 ,E_2)$ be graphs on disjoint sets
of vertices. Their \emph{disjoint union} $G_1+G_2$ is the graph $G_1+G_2=(V_1\cup V_2 , E_1\cup E_2).$
A \emph{coalescence} of graphs $G_1$ and $G_2$ is any graph obtained from
the disjoint union $G_1+G_2$ by identifying a vertex $v_1$ of $G_1$ with a vertex $v_2$ of $G_2$, that is, merging $v_1, v_2$ into
a single vertex $v$. If $v_1$, and $v_2$ have loops of weights $\alpha_1, \alpha_2,$ respectively, then in the coalescence $v$ will have a loop of weight $(\alpha_1+\alpha_2)$.

A \emph{block} in a graph is a maximal connected subgraph that has no cut-vertex (\cite{west2001introduction}, p.15). Note that if a connected graph has no cut-vertex, then it itself is a block. The blocks and cut-vertices in any graph can found in linear time \cite{hopcroft1971efficient}. A block in a graph $G$ is called a \emph{pendant block} if it has one cut-vertex of $G$. Two blocks are called \emph{adjacent blocks} if they share a cut-vertex of $G$. A complete graph on $n$ vertices is denoted by $K_n$. If every block of a graph is a complete graph then it is called a \emph{block graph}. An example of block graph is given in Figure \ref{r2fig}(a), where the blocks are the induced subgraph on the vertex-sets $\{3, 5, 6, 7\}, \{1, 4, 3, 2\}, \{2, {11}, {12}\}$, and $\{1, {10}, 8, 9\}$. For more details on the block graphs see (\cite{bapat2014graphs}, Chapter 7). 

A well-known problem, proposed in 1957, by Collatz and Sinogowitz, is to characterize graphs
with positive nullity \cite{von1957spektren}.  Nullity of graphs is applicable in various branches of science, in particular, quantum chemistry, H\"{u}ckel
molecular orbital theory \cite{gutman2011nullity,lee1994chemical} and social networks theory \cite{leskovec2010signed}.  There has been significant work on the nullity of undirected graphs like trees \cite{cvetkovic1980spectra,fiorini2005trees,gutman2001nullity},
unicyclic graphs \cite{hu2008nullity,nath2007null,hu2008nullity}, and bicyclic graphs \cite{bapat2011note,berman2008upper,xuezhong2005nullity}. It is well-known that a tree is nonsingular if and only if it has a perfect matching. Since a tree is a block graph, it is natural to investigate in general which block graphs are nonsingular. The combinatorial formulae for the determinant of block graphs are known in terms of size and arrangements of its blocks \cite{bapat2014adjacency,singh2017characteristic,singh2017mathcal}. But, due to a myriad of the possibility of sizes and arrangements of the blocks it is difficult in general to classify nonsingular or singular block graphs. In \cite{singh2018nonsingular}, some classes of singular and nonsingular block graphs are given. However, it is still an open problem. In this paper, we give a linear time algorithm to determine whether a given graph is singular or not.

Thanks to the eminent computer scientist Alan Turing, the matrix determinant can be computed in the polynomial time using the $LUP$ decomposition of the matrix, where $L$ is a lower triangular matrix, $U$ is an upper triangular matrix, $P$ is a permutation matrix. Interestingly, the asymptotic complexity of the matrix determinant is the same as that of matrix multiplication of two matrices of the same order. The theorem which relates the complexity of matrix product and the matrix determinant is as follows. 

\begin{theorem}\cite{aho1974design} \label{md}
Let $M(n)$ be the time required to multiply two $n\times n$ matrices over some ring, and $A$ be an $n\times n$ matrix. Then, we can compute the determinant of $A$ in $O(M(n))$ steps.
\end{theorem}

Rank, nullity is also calculated using $LUP$ decomposition, hence their complexities are also the same as that of the matrix determinant or product of two matrices of the same order. The fastest matrix product takes time of $O(n^{2.373})$ \cite{davie2013improved,williams2011breaking,le2014powers}.
 In this article, we give a linear time algorithm to check whether a given block graph $G$ of order $n$, is singular or nonsingular (full rank). The algorithm is linear in $n$. We will not make use of matrix product or elementary row/column operations during any stage of the computation, instead, we prove first results using elementary  row/column operation on matrices and use them to treat pendant blocks of the block graphs.

The rest of the article is organised as follows. In Section \ref{NP}, we mention some notations and preliminaries used in the article. In Section \ref{TA}, we give necessary and sufficient condition for any block graph to be nonsingular by  using prior results of elementary row/column operation  given in  Section \ref{NP}. In Section \ref{algo}, we provide an outline for Algorithm \ref{balgo} and provide some examples in support of it. A proof of the correctness of  Algorithm \ref{balgo} is given in Section \ref{pc}.

\section{Notations and preliminary results}\label{NP}
If $Q$ is a subgraph of graph $G$, then $G \setminus Q$ denotes the induced subgraph of $G$ on the vertex subset
$V(G) \setminus V(Q)$. Here, $V(G) \setminus V(Q)$ is the standard set-theoretic subtraction of vertex sets. If $Q$ consists of a 
single vertex $v$, then we write $G\setminus v$
for $G\setminus Q$.  $J,j,O,o$ denote the all-one matrix, all-one column vector, zero matrix, zero column vector of suitable order, respectively. $w$ denotes a $(0,1)$-column vector of suitable
 order. $diag(x_1,\hdots,x_n)$ denotes the diagonal matrix of order $n$, where the $i$-th diagonal entry is $x_i, i=1,\hdots,n$. If a graph has no vertices (hence no edges) we call it a \emph{void} graph.

\begin{theorem}\label{mtype}
Consider a matrix \begin{equation}
M=\begin{bmatrix}
x_1 & 1 & \hdots & 1\\ 1 & x_2 &\ddots &\vdots \\ \vdots & \ddots & \ddots  & 1\\1 &\hdots & 1& x_n  
\end{bmatrix}=J-D,
\end{equation} where $D=diag(1-x_1,\hdots,1-x_n).$

\begin{enumerate}
\item If exactly one of $x_1,\hdots, x_n$ is equal to 1, then $M$ is nonsingular.
\item If any two (or more) of $x_1,\hdots, x_n$ are equal to 1, then $M$ is singular.
\item  If $x_i\neq 1, i=1,\hdots,n$. Let $$S=\sum_{i=1}^n\frac{1}{1-x_i}.$$  Then 
\begin{enumerate}
\item $M$ is nonsingular if and only if $S \neq 1$.
\item if $S=1$ and $x_{n+1}$ is any real number, then the matrix

$$\begin{bmatrix}
M & j\\j^T & x_{n+1}
\end{bmatrix}$$ is nonsingular.  \end{enumerate}
\item If $M$ is nonsingular, then the matrix $$\tilde{M}=\begin{bmatrix}
 M & j \\ j^T & \alpha
 \end{bmatrix}$$ can be transformed to the following matrix using elementary row and column operations $$\begin{bmatrix}
 M & o \\ o^T & \alpha+\gamma
\end{bmatrix},$$ where
 \begin{equation}\label{valueofgamma}
 \gamma= \begin{cases}
 -\frac{S}{S-1} & \mbox{if $x_i\neq 1, i=1,\hdots n$},\\
 -1 & \mbox {if exactly one of $x_1,\hdots,x_n$ is equal to 1.}
\end{cases}  
 \end{equation}
\end{enumerate}
\end{theorem}
\begin{proof}
\begin{enumerate}
\item  Without loss of generality, let $x_1=1$. By subtracting the 1-st row from the rest of the rows, $M$ can be transformed to the matrix,
$$\begin{bmatrix}
1 & 1 & \hdots & 1\\ 0 & (x_2-1) &\ddots &\vdots \\ \vdots & \ddots & \ddots  & 1\\0 &\hdots & 0& (x_n-1)  
\end{bmatrix},$$ whose determinant $\prod_{i=2}^n(x_i-1)$ is nonzero.
\item In this case two rows (or columns) are the same.
\item \begin{enumerate}
\item  We have 
$$D-J=D^{1/2}(I-D^{-1/2}JD^{-1/2})D^{1/2},$$ 
$$D^{-1/2}JD^{-1/2}=D^{-1/2}jj^TD^{-1/2}=\big(D^{-1/2}j\big)\big(D^{-1/2}j\big)^T.$$
Let $y=D^{-1/2}j.$ The $J-D=M$ is nonsingular if and only if $I-yy^T$ is nonsingular. Since the eigenvalues of $yy^T$ are $||y||^2=y^Ty$ and $0,$ the eigenvalues of the matrix $I-yy^T$ are $1-||y||^2$ and $1$ hence it is nonsingular if and only if $||y||^2\neq 1.$ That is $M$ is nonsingular if and only if $$\sum_{i=1}^n\frac{1}{1-x_i}\neq 1.$$ 
\item \begin{enumerate}
\item If $x_{n+1}\neq 1:$ as $S=1,$ this implies $$S+\frac{1}{1-x_{n+1}}\neq 1.$$ Hence the result follows by 3(a).

\item If $x_{n+1}= 1:$ then result follows by 1.
\end{enumerate}
\end{enumerate} 
\item \begin{enumerate}
\item If $x_i\neq 1, i=1,\hdots n$. Let $y=D^{-1/2}j.$ We have, $$(I-yy^T)(I+tyy^T)=I+(t-1-t||y||^2)yy^T,$$ and therefore if $||y||^2\neq 1,$ $$(I-yy^T)^{-1}=I+\frac{1}{1-||y||^2}yy^T.$$ Thus if $J-D$ is invertible, where
$d_i=1-x_i,$ then, by the above, $$(J-D)^{-1}=-(D-J)^{-1}=-D^{-1/2}\Big(I+\frac{1}{1-||y||^2}yy^T\Big)D^{-1/2},$$ where $y=D^{-1/2}j.$

 \begin{align*} \label{vog}
j^TM^{-1}j &=-j^TD^{-1/2}\Big(I+\frac{1}{1-||y||^2}yy^T\Big)D^{-1/2}j\nonumber\\ &= -y^T\Big(I+\frac{1}{1-||y||^2}yy^T\Big)y=\frac{||y||^2}{||y||^2-1}.
\end{align*}

 Let
\begin{equation*}\label{fjds}
P=\begin{bmatrix}
I & -M^{-1}j\\ o^T & 1
\end{bmatrix}.
\end{equation*} 
Then
\begin{equation*}\label{fjde}
P^T\tilde{M}P=\begin{bmatrix}
M & o\\ o^T & \alpha+\gamma
\end{bmatrix},
\end{equation*}
where  $$\gamma=-j^TM^{-1}j=-\frac{||y||^2}{||y||^2-1}=-\frac{S}{S-1}.$$
\item When exactly one of  $x_1,\hdots,x_n$ is equal to 1, without loss of generality, $x_1=1,$ we can write
\begin{equation*}
\tilde{M}=\begin{bmatrix}
1 & 1 & \hdots & 1 & 1\\ 1 & x_2 &\ddots &\vdots & \vdots \\ \vdots & \ddots & \ddots  & 1 & 1\\1 &\hdots & 1& x_n & 1\\ 1 &\hdots & 1& 1& \alpha
\end{bmatrix}
\end{equation*}

On subtracting the first column from the last column, and subsequently subtracting  the first row from the last row. We get the following matrix, 
\begin{equation*}
\tilde{M}=\begin{bmatrix}
1 & 1 & \hdots & 1 & 0\\ 1 & x_2 &\ddots &\vdots & \vdots \\ \vdots & \ddots & \ddots  & 1 & 0\\1 &\hdots & 1& x_n & 0\\ 0 &\hdots & 0& 0& \alpha-1
\end{bmatrix},
\end{equation*} hence $\gamma=-1.$
\end{enumerate}
\end{enumerate}
\end{proof}

\begin{lemma} {\cite[Lemma 2.3]{singh2017characteristic}} \label{supfirst}
If $G$ is a coalescence of $G_1$ and $G_2$ at a vertex $v$ having loop of weight $\alpha$, then $$\det(G)=\det(G_1)\det(G_2\setminus v)+\det(G_1\setminus v)\det(G_2)-\alpha \det(G_1 \setminus v)\det(G_2\setminus v).$$
\end{lemma}

\begin{figure}

     \begin{subfigure}[b]{0.20\textwidth}
\begin{center}
 \begin{tikzpicture} [scale=0.7] [->,>=stealth',shorten >=1pt,auto,node distance=4cm,
                thick,main node/.style={circle,draw,font=\Large\bfseries}] 
               
 \draw  node[draw,circle,scale=0.40] (1) at (0,0) {$1$}; 
\draw  node[draw,circle,scale=0.40] (2) at (2,0) {$2$};
\draw  node[draw,circle,scale=0.40] (3) at (2,2) {$3$};
\draw  node[draw,circle,scale=0.40] (4) at (0,2) {$4$};
\draw  node[draw,circle,scale=0.40] (8) at (0,-2) {${8}$};
\draw  node[draw,circle,scale=0.40] (9) at (.75,-1.5) {${9}$}; 
\draw  node[draw,circle,scale=0.35] (10) at (-0.75,-1.5) {${10}$};

\draw  node[draw,circle,scale=0.40] (6) at (2,4) {${6}$};
\draw  node[draw,circle,scale=0.40] (7) at (2.75,3.5) {${7}$}; 
\draw  node[draw,circle,scale=0.40] (5) at (1.25,3.5) {${5}$}; 

\draw  node[draw,circle,scale=0.35] (11) at (2.85,0.65) {${11}$};
\draw  node[draw,circle,scale=0.35] (12) at (2.85,-0.65) {$12$};

\tikzset{edge/.style = {- = latex'}}
\draw[edge] (1) to (2);                  
 \draw[edge] (1) to (3);                  
 \draw[edge] (1) to (4);                  
 \draw[edge] (3) to (2);                  
 \draw[edge] (4) to (2);                  
 \draw[edge] (3) to (4);
 
 \draw[edge] (1) to (8);                  
 \draw[edge] (1) to (9);                  
 \draw[edge] (1) to (10);                  
 \draw[edge] (8) to (9);                  
 \draw[edge] (8) to (10);                  
 \draw[edge] (9) to (10);

\draw[edge] (3) to (5);                  
 \draw[edge] (3) to (6);                  
 \draw[edge] (3) to (7);                  
 \draw[edge] (5) to (6);                  
 \draw[edge] (5) to (7);                  
 \draw[edge] (6) to (7);  

\draw[edge] (2) to (11);                  
\draw[edge] (2) to (12);                  
\draw[edge] (11) to (12);

\end{tikzpicture}
     \caption{} \end{center} \label{31block}
    \end{subfigure}    \begin{subfigure}[b]{0.25\textwidth}
\begin{center} \begin{tikzpicture}[scale=0.7] [->,>=stealth',shorten >=1pt,auto,node distance=4cm,
                thick,main node/.style={circle,draw,font=\Large\bfseries}]
                       \tikzset{vertex/.style = {shape=circle,draw}}
 \tikzset{edge/.style = {- = latex'}}
                                         vertices
\draw  node[draw,circle,scale=0.40] (1) at (0,0) {$1$}; 
\draw  node[draw,circle,scale=0.40] (2) at (2,0) {$2$};
\draw  node[draw,circle,scale=0.40] (3) at (2,2) {$3$};
\draw  node[draw,circle,scale=0.40] (4) at (0,2) {$4$};
\draw  node[draw,circle,scale=0.40] (8) at (0,-2) {${8}$};
\draw  node[draw,circle,scale=0.40] (9) at (.75,-1.5) {${9}$}; 
\draw  node[draw,circle,scale=0.40] (10) at (-0.75,-1.5) {${10}$}; 

\draw  node[draw,circle,scale=0.35] (11) at (2.85,0.65) {${11}$};
\draw  node[draw,circle,scale=0.35] (12) at (2.85,-0.65) {${12}$};

\draw  node[draw,circle,scale=0.40] (6) at (2,4) {${6}$};
\draw  node[draw,circle,scale=0.40] (7) at (2.75,3.5) {${7}$}; 
\draw  node[draw,circle,scale=0.40] (5) at (1.25,3.5) {${5}$};

\tikzset{edge/.style = {- = latex'}}
\draw[edge] (1) to (2);                  
 \draw[edge] (1) to (3);                  
 \draw[edge] (1) to (4);                  
 \draw[edge] (3) to (2);                  
 \draw[edge] (4) to (2);                  
 \draw[edge] (3) to (4);
 
 \draw[edge] (1) to (8);                  
 \draw[edge] (1) to (9);                  
 \draw[edge] (1) to (10);                  
 \draw[edge] (8) to (9);                  
 \draw[edge] (8) to (10);                  
 \draw[edge] (9) to (10);
\draw[edge] (2) to (11);                  
\draw[edge] (2) to (12);                  
\draw[edge] (11) to (12);

 \draw[edge] (5) to (6);                  
 \draw[edge] (5) to (7);                  
 \draw[edge] (6) to (7);
 \node[] at (2.25,2.3) {$-\frac{3}{2}$};                  
 
\draw (3) to [out=330,in=300,looseness=12] (3); 
   
 \end{tikzpicture}
  \caption{} 
  \end{center}\label{r2biblock}
    \end{subfigure} \begin{subfigure}[b]{0.25\textwidth}
 \begin{center}
\begin{tikzpicture}[scale=0.7] [->,>=stealth',shorten >=1pt,auto,node distance=4cm,
                thick,main node/.style={circle,draw,font=\Large\bfseries}]
                       \tikzset{vertex/.style = {shape=circle,draw}}
 \tikzset{edge/.style = {- = latex'}}
                                         vertices
\draw  node[draw,circle,scale=0.40] (1) at (0,0) {$1$}; 
\draw  node[draw,circle,scale=0.40] (2) at (2,0) {$2$};
\draw  node[draw,circle,scale=0.40] (3) at (2,2) {$3$};
\draw  node[draw,circle,scale=0.35] (11) at (2.85,0.65) {${11}$};
\draw  node[draw,circle,scale=0.35] (12) at (2.85,-0.65) {${12}$};

\node[] at (2.25,2.3) {$-\frac{3}{2}$};                  

\node[] at (-0.2,-0.6) {$-\frac{3}{2}$};  
\draw  node[draw,circle,scale=0.40] (4) at (0,2) {$4$};
\draw  node[draw,circle,scale=0.40] (8) at (0,-2) {${8}$};
\draw  node[draw,circle,scale=0.40] (9) at (.75,-1.5) {${9}$}; 
\draw  node[draw,circle,scale=0.40] (10) at (-0.75,-1.5) {${10}$};

\draw  node[draw,circle,scale=0.40] (6) at (2,4) {${6}$};
\draw  node[draw,circle,scale=0.40] (7) at (2.75,3.5) {${7}$}; 
\draw  node[draw,circle,scale=0.40] (5) at (1.25,3.5) {${5}$};

\tikzset{edge/.style = {- = latex'}}
 \draw[edge] (3) to (2);                  
 \draw[edge] (4) to (2);                  
 \draw[edge] (3) to (4);

 \draw[edge] (8) to (9);                  
 \draw[edge] (8) to (10);                  
 \draw[edge] (9) to (10);

\draw[edge] (2) to (11);                  
\draw[edge] (2) to (12);                  
\draw[edge] (11) to (12);

\draw[edge] (1) to (3);                  
\draw[edge] (2) to (1);                  
\draw[edge] (4) to (1);

 \draw[edge] (5) to (6);                  
 \draw[edge] (5) to (7);                  
 \draw[edge] (6) to (7);  
 \draw (3) to [out=330,in=300,looseness=12] (3);

 \draw (1) to [out=330,in=300,looseness=12] (1);
 \end{tikzpicture}
  \caption{}\end{center}  
  \label{r3biblock}
    \end{subfigure}  \begin{subfigure}[b]{0.25\textwidth}
 \begin{center}
\begin{tikzpicture}[scale=0.7] [->,>=stealth',shorten >=1pt,auto,node distance=4cm,
                thick,main node/.style={circle,draw,font=\Large\bfseries}]
                       \tikzset{vertex/.style = {shape=circle,draw}}
 \tikzset{edge/.style = {- = latex'}}
                                         vertices
\draw  node[draw,circle,scale=0.40] (1) at (0,0) {$1$}; 
\draw  node[draw,circle,scale=0.40] (2) at (2,0) {$2$};
\draw  node[draw,circle,scale=0.40] (3) at (2,2) {$3$};
\draw  node[draw,circle,scale=0.35] (11) at (2.85,0.65) {${11}$};
\draw  node[draw,circle,scale=0.35] (12) at (2.85,-0.65) {${12}$};

\node[] at (2.25,2.3) {$-\frac{3}{2}$};                  

\node[] at (-0.2,-0.6) {$-\frac{3}{2}$};

\node[] at (1.35,-0.35) {$-\frac{9}{4}$};
  
\draw  node[draw,circle,scale=0.40] (4) at (0,2) {$4$};
\draw  node[draw,circle,scale=0.40] (8) at (0,-2) {${8}$};
\draw  node[draw,circle,scale=0.40] (9) at (.75,-1.5) {${9}$}; 
\draw  node[draw,circle,scale=0.40] (10) at (-0.75,-1.5) {${10}$};

\draw  node[draw,circle,scale=0.40] (6) at (2,4) {${6}$};
\draw  node[draw,circle,scale=0.40] (7) at (2.75,3.5) {${7}$}; 
\draw  node[draw,circle,scale=0.40] (5) at (1.25,3.5) {${5}$};

\tikzset{edge/.style = {- = latex'}}

 \draw[edge] (3) to (4);

 \draw[edge] (8) to (9);                  
 \draw[edge] (8) to (10);                  
 \draw[edge] (9) to (10);

\draw[edge] (2) to (11);                  
\draw[edge] (2) to (12);                  
\draw[edge] (11) to (12);

\draw[edge] (1) to (3);                  
\draw[edge] (4) to (1);

 \draw[edge] (5) to (6);                  
 \draw[edge] (5) to (7);                  
 \draw[edge] (6) to (7);  
 \draw (3) to [out=330,in=300,looseness=12] (3);

 \draw (1) to [out=330,in=300,looseness=12] (1);
 \draw (2) to [out=270,in=300,looseness=12] (2);
 \end{tikzpicture}
  \caption{}\end{center}  
  \label{r3biblock}
    \end{subfigure} 
\caption{Results of elementary  row/column operations on the succesive pendant blocks of the block graph in (a). In (d) all the four weighted induced subgraphs are nonsingular, hence, the block graph in (a) is nonsingular.} \label{r2fig}
\end{figure}
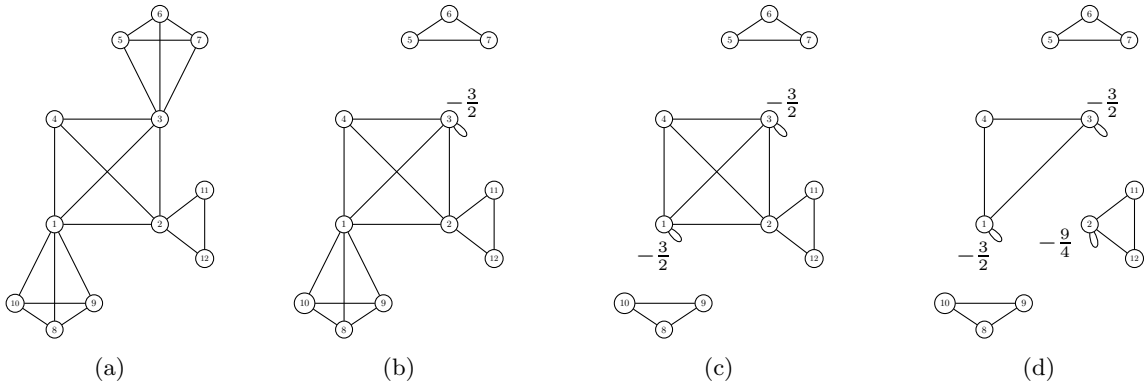

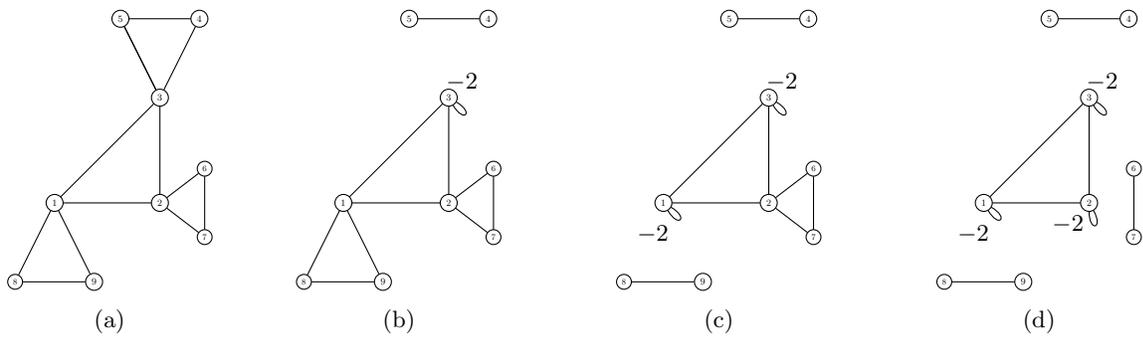
\begin{figure}

     \begin{subfigure}[b]{0.20\textwidth}
\begin{center}
 \begin{tikzpicture} [scale=0.7] [->,>=stealth',shorten >=1pt,auto,node distance=4cm,
                thick,main node/.style={circle,draw,font=\Large\bfseries}] 
               
 \draw  node[draw,circle,scale=0.40] (1) at (0,0) {$1$}; 
\draw  node[draw,circle,scale=0.40] (2) at (2,0) {$2$};
\draw  node[draw,circle,scale=0.40] (3) at (2,2) {$3$};
\draw  node[draw,circle,scale=0.40] (9) at (.75,-1.5) {${9}$}; 
\draw  node[draw,circle,scale=0.35] (8) at (-0.75,-1.5) {$8$};

\draw  node[draw,circle,scale=0.40] (4) at (2.75,3.5) {${4}$}; 
\draw  node[draw,circle,scale=0.40] (5) at (1.25,3.5) {${5}$}; 

\draw  node[draw,circle,scale=0.35] (6) at (2.85,0.65) {$6$};
\draw  node[draw,circle,scale=0.35] (7) at (2.85,-0.65) {$7$};

\tikzset{edge/.style = {- = latex'}}
\draw[edge] (1) to (2);                  
 \draw[edge] (1) to (3);                  
 \draw[edge] (3) to (4);                  
 \draw[edge] (3) to (2);                  
 \draw[edge] (4) to (5);                  
 \draw[edge] (3) to (5);
 
 \draw[edge] (1) to (8);                  
 \draw[edge] (1) to (9);                  
 \draw[edge] (8) to (9);                  
 
\draw[edge] (3) to (5);                  
                   
 \draw[edge] (2) to (6);                  
 \draw[edge] (2) to (7);                  
 \draw[edge] (6) to (7);

\end{tikzpicture}
     \caption{} \end{center} \label{31block}
    \end{subfigure}    \begin{subfigure}[b]{0.25\textwidth}
\begin{center} \begin{tikzpicture}[scale=0.7] [->,>=stealth',shorten >=1pt,auto,node distance=4cm,
                thick,main node/.style={circle,draw,font=\Large\bfseries}]
                       \tikzset{vertex/.style = {shape=circle,draw}}
 \tikzset{edge/.style = {- = latex'}}
                                         vertices
\draw  node[draw,circle,scale=0.40] (1) at (0,0) {$1$}; 
\draw  node[draw,circle,scale=0.40] (2) at (2,0) {$2$};
\draw  node[draw,circle,scale=0.40] (3) at (2,2) {$3$};
\draw  node[draw,circle,scale=0.40] (9) at (.75,-1.5) {${9}$}; 
\draw  node[draw,circle,scale=0.35] (8) at (-0.75,-1.5) {$8$};

\draw  node[draw,circle,scale=0.40] (4) at (2.75,3.5) {${4}$}; 
\draw  node[draw,circle,scale=0.40] (5) at (1.25,3.5) {${5}$}; 

\draw  node[draw,circle,scale=0.35] (6) at (2.85,0.65) {$6$};
\draw  node[draw,circle,scale=0.35] (7) at (2.85,-0.65) {$7$};

\tikzset{edge/.style = {- = latex'}}
\draw[edge] (1) to (2);                  
 \draw[edge] (1) to (3);                  
                  
 \draw[edge] (3) to (2);                  
 \draw[edge] (4) to (5);

 \draw[edge] (1) to (8);                  
 \draw[edge] (1) to (9);                  
 \draw[edge] (8) to (9);

 \draw[edge] (2) to (6);                  
 \draw[edge] (2) to (7);                  
 \draw[edge] (6) to (7);  

\draw (3) to [out=330,in=300,looseness=12] (3);

\node[] at (2.25,2.3) {$-2$};

\end{tikzpicture}
  \caption{} 
  \end{center}\label{r2biblock}
    \end{subfigure} \begin{subfigure}[b]{0.25\textwidth}
 \begin{center}
\begin{tikzpicture}[scale=0.7] [->,>=stealth',shorten >=1pt,auto,node distance=4cm,
                thick,main node/.style={circle,draw,font=\Large\bfseries}]
                       \tikzset{vertex/.style = {shape=circle,draw}}
 \tikzset{edge/.style = {- = latex'}}
                                         vertices
\draw  node[draw,circle,scale=0.40] (1) at (0,0) {$1$}; 
\draw  node[draw,circle,scale=0.40] (2) at (2,0) {$2$};
\draw  node[draw,circle,scale=0.40] (3) at (2,2) {$3$};
\draw  node[draw,circle,scale=0.40] (9) at (.75,-1.5) {${9}$}; 
\draw  node[draw,circle,scale=0.35] (8) at (-0.75,-1.5) {$8$}; 
\node[] at (-0.2,-0.6) {$-2$};

\draw  node[draw,circle,scale=0.40] (4) at (2.75,3.5) {${4}$}; 
\draw  node[draw,circle,scale=0.40] (5) at (1.25,3.5) {${5}$}; 

\draw  node[draw,circle,scale=0.35] (6) at (2.85,0.65) {$6$};
\draw  node[draw,circle,scale=0.35] (7) at (2.85,-0.65) {$7$};

\tikzset{edge/.style = {- = latex'}}
\draw[edge] (1) to (2);                  
 \draw[edge] (1) to (3);                  
                  
 \draw[edge] (3) to (2);                  
 \draw[edge] (4) to (5);

 \draw[edge] (8) to (9);

 \draw[edge] (2) to (6);                  
 \draw[edge] (2) to (7);                  
 \draw[edge] (6) to (7);  

\draw (3) to [out=330,in=300,looseness=12] (3);

\node[] at (2.25,2.3) {$-2$};

 \draw (1) to [out=330,in=300,looseness=12] (1);
 \end{tikzpicture}
  \caption{}\end{center}  
  \label{r3biblock}
    \end{subfigure}  \begin{subfigure}[b]{0.25\textwidth}
 \begin{center}
\begin{tikzpicture}[scale=0.7] [->,>=stealth',shorten >=1pt,auto,node distance=4cm,
                thick,main node/.style={circle,draw,font=\Large\bfseries}]
                       \tikzset{vertex/.style = {shape=circle,draw}}
 \tikzset{edge/.style = {- = latex'}}
                                         vertices
\draw  node[draw,circle,scale=0.40] (1) at (0,0) {$1$}; 
\draw  node[draw,circle,scale=0.40] (2) at (2,0) {$2$};
\draw  node[draw,circle,scale=0.40] (3) at (2,2) {$3$};
\draw  node[draw,circle,scale=0.40] (9) at (.75,-1.5) {${9}$}; 
\draw  node[draw,circle,scale=0.35] (8) at (-0.75,-1.5) {$8$}; 
\node[] at (-0.2,-0.6) {$-2$};

\node[] at (1.60,-0.4) {$-2$};                  

\draw  node[draw,circle,scale=0.40] (4) at (2.75,3.5) {${4}$}; 
\draw  node[draw,circle,scale=0.40] (5) at (1.25,3.5) {${5}$}; 

\draw  node[draw,circle,scale=0.35] (6) at (2.85,0.65) {$6$};
\draw  node[draw,circle,scale=0.35] (7) at (2.85,-0.65) {$7$};

\tikzset{edge/.style = {- = latex'}}
\draw[edge] (1) to (2);                  
 \draw[edge] (1) to (3);                  
                  
 \draw[edge] (3) to (2);                  
 \draw[edge] (4) to (5);

 \draw[edge] (8) to (9);

 \draw[edge] (6) to (7);  

\draw (3) to [out=330,in=300,looseness=12] (3);

\node[] at (2.25,2.3) {$-2$};

 \draw (1) to [out=330,in=300,looseness=12] (1);
 \draw (2) to [out=270,in=300,looseness=12] (2);
 \end{tikzpicture}
  \caption{}\end{center}  
  \label{r3biblock}
    \end{subfigure} 
\caption{Results of elementary  row/column operations on the succesive pendant blocks of the block graph in (a). In (d) the weighted induced subgraph on vertices $1,2,3$ is singular, hence, the block graph in (a) is singular.} \label{r3fig}
\end{figure}  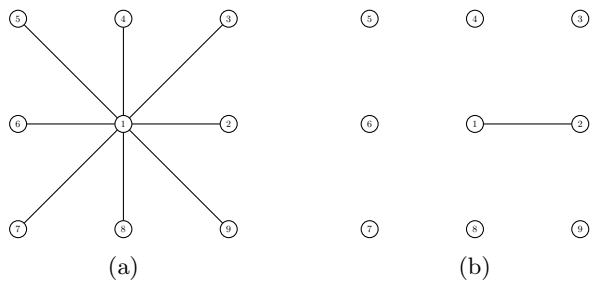
\begin{figure} \centering

     \begin{subfigure}[b]{0.20\textwidth}
\begin{center}
 \begin{tikzpicture} [scale=0.7] [->,>=stealth',shorten >=1pt,auto,node distance=4cm,
                thick,main node/.style={circle,draw,font=\Large\bfseries}] 
               
 \draw  node[draw,circle,scale=0.40] (1) at (0,0) {$1$}; 
\draw  node[draw,circle,scale=0.40] (2) at (2,0) {$2$};
\draw  node[draw,circle,scale=0.40] (3) at (2,2) {$3$};

\draw  node[draw,circle,scale=0.40] (4) at (0,2) {$4$};

\draw  node[draw,circle,scale=0.40] (5) at (-2,2) {$5$};
\draw  node[draw,circle,scale=0.40] (6) at (-2,0) {$6$};
\draw  node[draw,circle,scale=0.40] (7) at (-2,-2) {$7$};
\draw  node[draw,circle,scale=0.40] (8) at (0,-2) {$8$};

\draw  node[draw,circle,scale=0.40] (9) at (2,-2) {$9$};

\tikzset{edge/.style = {- = latex'}}
\draw[edge] (1) to (2);                  
 \draw[edge] (1) to (3);                  
 \draw[edge] (1) to (4);                  
                  
 \draw[edge] (1) to (5);                  
 \draw[edge] (1) to (6);
 \draw[edge] (1) to (7);
 \draw[edge] (1) to (8);                  
 \draw[edge] (1) to (9);                  
 
\end{tikzpicture}
     \caption{} \end{center} \label{31block}
    \end{subfigure}    \begin{subfigure}[b]{0.35\textwidth}
\begin{center} \begin{tikzpicture}[scale=0.7] [->,>=stealth',shorten >=1pt,auto,node distance=4cm,
                thick,main node/.style={circle,draw,font=\Large\bfseries}]
                       \tikzset{vertex/.style = {shape=circle,draw}}
 \tikzset{edge/.style = {- = latex'}}
                                         vertices
 \draw  node[draw,circle,scale=0.40] (1) at (0,0) {$1$}; 
\draw  node[draw,circle,scale=0.40] (2) at (2,0) {$2$};
\draw  node[draw,circle,scale=0.40] (3) at (2,2) {$3$};

\draw  node[draw,circle,scale=0.40] (4) at (0,2) {$4$};

\draw  node[draw,circle,scale=0.40] (5) at (-2,2) {$5$};
\draw  node[draw,circle,scale=0.40] (6) at (-2,0) {$6$};
\draw  node[draw,circle,scale=0.40] (7) at (-2,-2) {$7$};
\draw  node[draw,circle,scale=0.40] (8) at (0,-2) {$8$};

\draw  node[draw,circle,scale=0.40] (9) at (2,-2) {$9$};

\tikzset{edge/.style = {- = latex'}}
\draw[edge] (1) to (2);  
\end{tikzpicture}
  \caption{} 
  \end{center}\label{r2biblock}
    \end{subfigure}  
\caption{Results of elementary  row/column operations on the succesive pendant blocks of the block graph in (a). In (d) there are isolated vertices with zero weights, hence, the block graph in (a) is singular.} \label{r4fig}
\end{figure}

\section{Transformation of the adjacency matrix} \label{TA}
Elementary  row/column operations transform the adjacency matrix $A(G)$ into another matrix $EA(G),$
which may have nonzero diagonal entries.  The graph corresponding to $EA(G)$ can be obtained by a change of weights on the edges and addition of loops on the vertices in $G$ according
to the elementary row/column operations on $A(G)$. This can be seen as follows. Let $B$ be a pendant block of a block graph $G$, having the cut-vertex $v$ of $G$. As reordering the vertices of $G$ does not change the rank of $G$, we can write the adjacency matrix $A(G)$ as the follows.
$$A(G)=\begin{bmatrix}
A(B\setminus v)& j & O^T\\ j^T & 0 & w^T \\ O & w & A(G\setminus B) 
\end{bmatrix}.$$ 
Let the pendant block $B$ have $m$ vertices. 
\begin{enumerate}
\item If $m \geq 3:$ as $A(B\setminus v)$ is the complete graph on $m-1 \geq 2$ vertices, it is  nonsingular. All the diagonal elements in $A(B\setminus v)$ are zero.  Using Theorem \ref{mtype}(4), $A(G)$ is transformed to the following matrix.
$$M_1=\begin{bmatrix}
A(B\setminus v)& o & O^T\\ o^T & -\frac{m-1}{m-2} & w^T \\ O & w & A(G\setminus B) 
\end{bmatrix}. $$ Let $$M_{12}= \begin{bmatrix}
  -\frac{m-1}{m-2} & w^T \\ w & A(G\setminus B) 
\end{bmatrix}.$$ We have,
$r(A(G))=r(M_1)=r(A(B\setminus v))+r(M_{12}).$ As $A(B\setminus v)$ is nonsingular,  $A(G)$ is nonsingular if and only if $M_{12}$ is nonsingular.

\item If $m=2:$ In this case $G$ is a coalescence of $B=K_2$ with $G\setminus(B\setminus v)$ at the vertex $v$. As $K_2$ is nonsingular and $K_2 \setminus v$ is singular, using Lemma \ref{supfirst}, $G$ is nonsingular if and only if $A(G\setminus B)$ is nonsingular. 
\end{enumerate}

In  view of the above discussion, in order to know whether block graph $G$ is singular or nonsingular, we need to further examine the matrices $M_{12}$ for $m\geq 3$, $A(G\setminus B),$ for $m=2$. Note that, $M_{12}$, is the matrix corresponding to the induced subgraph on the vertex-set $V(G)\setminus V(B\setminus v)$ of $G$, with a loop on the vertex $v$ having weight $-\frac{m-1}{m-2}$. By selecting a pendant block from the graph $G\setminus B$ (for  $m=2$) or the graph corresponding to $M_{12}$ ( for $m\geq 3$)  we can further investigate the rank of $G$. This process will continue until we cover all the blocks of $G$. The matrix corresponding to the pendant block at any step is of the form $M$ in Theorem \ref{mtype}. Hence we use the prior results from Section \ref{NP} to check whether $G$ is singular or not. Examples are given in Figure \ref{r2fig}  (nonsingular), Figure  \ref{r3fig} (singular) and Figure  \ref{r4fig} (singular), where at each step a pendant block is chosen for the elementary  row/column operations. The detailed algorithm is given in the next section.

\section{Algorithm}  
\label{algo}

For the purpose of the algorithm we first define two auxiliary operations on a set of sets of integers. 
\subsection{Auxiliary operations}
Let $U=\{S_1, \hdots, S_k\},$ where, $S_i, i=1,\hdots,k,$ is a set of positive integers. We define two operations on $U:$ 
\begin{enumerate}
\item $U-S_i=\{S_j \mid j\neq i\}.$
\item $U-^\star S_i=\{S_j\setminus S_i \mid j\neq i\}.$
\end{enumerate}

\begin{example}\label{example1} 
Let $S_1=\{2,3,4,5\}, S_2=\{2,7,6,5\}, S_3=\{1,3,8,5\}$. Consider $U=\{S_1,S_2,S_3\}.$ Then,
\begin{enumerate}
\item $U-S_2=\{\{2,3,4,5\}, \{1,3,8,5\} \}.$
\item $U-^\star S_2=\{\{3,4\}, \{1,3,8\} \}.$
\end{enumerate}
\end{example}

The cardinality of a set $\mathcal{U}$ is denoted by $|\mathcal{U}|$. In Example \ref{example1}, the cardinality of $U$ is 3.

\begin{algorithm} 
\Fontvi
\SetAlgoLined
\KwResult{$G$ is \textbf{singular} or $G$ is \textbf{nonsingular.}

$BV=\{V(B_1),\hdots, V(B_k)\},$\\
$CV=\{C(B_1),\hdots, C(B_k)\},$ \\
$W=\{w(1),\hdots,w(n)\}, w(i)=0, i=1,\hdots,n,$\\ $f(i)$ is the number of times vertex $i$ appears in $CV$.}

\textbf{GSING} $(BV, CV)$\;
$m:=0;$ \Comment{$m$ is a variable to check the existence of a pendant vertex-set.}\\ 
\For{$i=1:|CV|$}{ \If{$(|CV(i)|=1)$}{    
   $SBV:=BV(i)$\; 
   $p:=CV(i)$\;
   $V:= SBV\setminus p$\; $m:=1$\; break\;}}

\eIf{$(m=1)$}{
 $t:=0;$ \Comment{$t$ is a variable to count the number of nonzero weights.}\\
\For{$i=1:|V|$}{ \If{$(w(V(i))\sim=1)$}{    
   $t:=t+1$\;} }
  \eIf{$(t\sim~=|V|)$}{    
  \eIf{$(t< |V|-1)$}{$G$ is \textbf{singular} \; exit;}{
  $w(p):=w(p)-1$\;
$BV=BV-SBV$\;  
\eIf{$(f(p)==2)$}{$CV:=CV-^\star p$\;}{
  $CV:=CV-p$\;
$f(p):=f(p)-1$\;  
}   } \textbf{GSING} $(BV, CV)$ }
   { $S:=0;$ \Comment{$S$ is a variable to calculate the sum as in Theorem \ref{mtype}.3.}\\ \For{$i=1:|V|$}{$S:=S+\frac{1}{1-w(V(i))}$\;}
  \eIf{$S\neq 1$}{ $w(p):=w(p)-\frac{S}{S-1}$\;
$BV:=BV-SBV$\; \eIf{$(f(p)==2)$}{$CV:=CV-^\star p$\;}{
  $CV:=CV-p$\;
$f(p):=f(p)-1$\;  
} }{
  $BV:=BV-^\star SBV$\; 
 $CV:=CV-^\star p$\;
} \textbf{GSING} $(BV, CV)$ 
  } }{ \For{$i=1:|BV|$ }{ \If{$BV(i)$ is singular}{  $G$ is \textbf{singular}\; break\;}{}}  {$G$ is \textbf{nonsingular}\;} 
} 
 \caption{Algorithm to check whether a block graph $G$ is singular or not.}\label{balgo}
\end{algorithm}

 \subsection{Algorithm}

Given a block graph $G$, with blocks $B_1,\hdots, B_k.$ We have a set $BV$ which is the set of vertex-sets of the blocks $B_1,\hdots, B_k,$ in $G$. $CV$ is the set of cut-vertex-sets of the blocks $B_1,\hdots, B_k.$ Let $V(B_i)$ denote the vertex set of the block $B_i$ and $C(B_i)$ be the set of cut-vertices in it. That is, we have
$$BV=\{V(B_1),\hdots, V(B_k)\},$$
$$CV=\{C(B_1),\hdots, C(B_k)\}.$$
A vertex-set $V(B_i)$ in $BV$ is \emph{pendant vertex-set} if $C(B_i)$ contains exactly one cut-vertex, that is, $|C(B_i)|=1$.
Let $f(p)$ be the number of times  the cut-vetex $p$ appears in $CV$. Let $w(i)$ be the weight assigned to the vertex $i,$ $i=1,\hdots,n$ of $G$. We first give an outline of the procedure to check whether a given block graph $G$ is singular or nonsingular. We start with the sets $BV, CV$. Initially, $w(i)=0, i=1,\hdots,n.$

\textbf{BEGIN-} If $BV$ has no pendant vertex-set, go to \textbf{END}. Otherwise, pick a pendant vertex-set $V_p$, where the cut-vertex is $p$. 

\textbf{Check} if $w(i)\neq 1, \forall i \in (V_p\setminus p).$
\begin{enumerate}
\item If false: \begin{enumerate}
\item If for more than one $i,$ $w(i)$ is 1, then $G$ is \textbf{singular}. \textbf{STOP}.
\item If for exactly one $i,$ $w(i)$ is 1, then induced subgraph on $V_p\setminus p$ is nonsingular. \textbf{DO}
\begin{enumerate}
\item $w(p):=w(p)-1.$
\item $BV:=BV-V_p$.  
\item If $f(p)=2,$ then $CV:=CV-^\star p,$ else,  $CV:=CV-p, f(p):=f(p)-1.$
\end{enumerate}
 \textbf{BEGIN} the algorithm for $BV, CV$.
\end{enumerate}

\item If true: Let $S=\sum_{i \in (V_p \setminus p)} \frac{1}{1-w(i)}$. \textbf{Check} If $ S \neq 1.$
\begin{enumerate}
\item If true:
Then the induced subgraph on $V_p\setminus p$ is nonsingular. \textbf{DO}
\begin{enumerate}
\item $w(p):=w(p)-\frac{S}{S-1}.$
\item $BV:=BV-V_p$.  
\item If $f(p)=2,$ $CV:=CV-^\star p,$ else,  $CV:=CV-p, f(p):=f(p)-1.$
\end{enumerate}
 \textbf{BEGIN} the algorithm for $BV, CV$.

\item If false: 
then induced subgraph on $V_p\setminus p$ is singular. \textbf{DO}
\begin{enumerate}
\item $BV:=BV-^\star V_p$.  
\item $CV:=CV-^\star p.$
\end{enumerate}
 \textbf{BEGIN} the algorithm for $BV, CV$.

\end{enumerate}
\end{enumerate}

\textbf{END:}  Let $H_1,\hdots, H_t$ be the induced subgraphs of $G$ on the vertex sets in $BV$ with the addition of weighted loops added in the process. If $H_i$ is nonsingular for all $i=1,\hdots, t,$ then $G$ is \textbf{nonsingular} else $G$ is \textbf{singular}. Note that, each of $H_i$ is of form $M$ in Theorem \ref{mtype}, hence, using Theorem \ref{mtype},  it is checked in linear time whether $H_i$ is singular or nonsingular.

The algorithm of the above outline is given in Algorithm \ref{balgo}. Where the function \textbf{GSING} takes $BV, CV$ as input, and makes recursive use of basic arithmetic, auxiliary operations which take at the most linear time, and check whether the block graph $G$ is singular or not. Now we provide some examples in support of Algorithm \ref{balgo}.

\begin{example}\label{ex1}
Consider the block graph given in Figure \ref{r2fig}(a). $$BV=\{\{3,5,6,7\}, \{1,10,8,9\},  \{3,4,1,2\}, \{2,11,12\} \},$$ 
$$CV=\{\{3\}, \{1\}, \{3,1,2\}, \{2\}\},$$
$f(3)=f(2)=f(1)=2.$ 
\begin{enumerate}
\item 1st call to $GSING(BV, CV)$. As 3 is the only vertex in $\{3\}$ in $CV$, the corresponding vertex-set $\{3,5,6,7\}$ in $BV$ is selected as pendant vertex-set. Next, induced graph on vertex set $\{3,5,6,7\} \setminus 3$ is investigated. $S=3$, as $w(5)=w(6)=w(7)=0$. Then $w(3)$ will be updated to $w(3)-\frac{3}{2}$, that is vertex $3$ will now have a loop of weight $-\frac{3}{2}$. Now,
$$BV=\{\{1,10,8,9\}, \{3,4,1,2\},  \{2,11,12\} \},$$ 
As, $f(3)=2,$ $$CV=\{\{1\}, \{1,2\}, \{2\}\}.$$

\item 2nd call to $GSING(BV, CV)$. As 1 is the only vertex in $\{1\}$ in $CV$, the corresponding vertex-set $\{1,10,8,9\}$ in $BV$ is selected as pendant vertex-set. Next, induced graph on vertex set $\{1,10,8,9\} \setminus 1$ is investigated. As $w(10)=w(8)=w(9)=0,$ $S=3$. Then $w(1)$ will be updated to $w(1)-\frac{3}{2}$, that is vertex $1$ will now have a loop of weight $-\frac{3}{2}$. 
$$BV=\{\{3,4,1,2\}, \{2,11,12\} \},$$ 
As, $f(1)=2,$ $$CV=\{\{2\}, \{2\}\}.$$

\item 3rd call to $GSING(BV, CV)$. As 2 is the only vertex in $\{2\}$ in $CV$, the corresponding vertex-set $\{3, 4, 1, 2\}$ in $BV$ is selected as pendant vertex-set. Next, induced graph on vertex set $\{3, 4, 1, 2\} \setminus 2$ is investigated. As $w(1)=-\frac{3}{2}, w(3)=-\frac{3}{2}, w(4)=0,$ $S=\frac{9}{5}$. Then $w(2)$ will be updated to $w(2)-\frac{9}{4}$, that is vertex $2$ will now have a loop of weight $-\frac{9}{4}$. 
$$BV=\{\{2,11,12\} \},$$ 
As, $f(1)=2,$ $$CV=\{ \}.$$

\item 4-th call to $GSING(BV, CV)$. As  $CV$ is empty, there is no pendant vertex-set. We need to investigate the induced graph $H_1$ on vertex set $\{2,11,12\}$ (with loop on 2). As $w(2)=-\frac{9}{4}, w(11)=0, w(12)=0,$ and hence $S=\frac{30}{13}\neq 1$. Which means $H_1$ is nonsingular, hence $G$ is nonsingular. 

\end{enumerate}
\end{example}

\begin{example}
The steps for block graph in Figure \ref{r3fig}(a) are similar to the steps in Example \ref{ex1} for block graph in Figure \ref{r2fig}(a). In the 4-th call to $GSING(BV, CV)$, we need to investigate the induced graph $H_1$ on vertex set $\{1, 2, 3\}$ (with loops). As $w(1)=w(2)=w(3)=-2,$ and hence $S=1$. Which means $H_1$ is singular, hence $G$ is singular. 
\end{example}

\begin{example}
 Consider block graph in Figure \ref{r4fig}(a). We have,
 $$BV=\{\{1,2\}, \{1,3\},  \{1,4\}, \{1,5\}, \{1,6\}, \{1,7\}, \{1,8\}, \{1,9\} \},$$ 
$$CV=\{\{1\}, \{1\}, \{1\}, \{1\}, \{1\}, \{1\}, \{1\}, \{1\}\},$$
$f(1)=2.$ 
\end{example}

\begin{enumerate}
\item 1st call to $GSING(BV, CV)$. As 1 is the only vertex in $\{1\}$ in $CV$, the corresponding vertex-set $\{1,2\}$ in $BV$ is selected as pendant vertex-set. Next, induced graph on vertex set $\{1,2\} \setminus 1$ is investigated. As $w(2)=0,$ $S=1$. Thus
$$BV=\{\{3\},  \{4\}, \{5\}, \{6\}, \{7\}, \{8\}, \{9\} \},$$ 
$$CV=\{ \{\}, \{\}, \{\}, \{\}, \{\}, \{\}, \{\}\},$$
 
\item 2nd call to $GSING(BV, CV)$. As  $CV$ is empty, there is no pendant vertex-set. We need to investigate the induced graphs on vertex-sets in $BV$. But as all are singleton vertices without loops. Hence the block graph is singular. 
\end{enumerate}

\section{Proof of correctness} \label{pc}
Just before $(i)$-th calling of function $GSING(BV, CV)$ in Algorithm \ref{balgo}, the adjacency matrix $A(G)$ is transformed to the matrix $M_{i}$ which is of the form 
\begin{equation*} M_i=
\begin{bmatrix}
A(G_1) & & & \\  &\ddots & &  \\ & & A(G_{i-1}) &  \\  & &  & A(G_{i} \setminus {i}) & j\\ & & & j^T & \alpha & w^T \\ & & & & w & M^{\star}
\end{bmatrix}.
\end{equation*} where each $G_j, \ j=1,\hdots, i-1,$ are some induced subgraph of block $B$  (possibly with loops) whose vertex-set selected as pendant vertex-set during $j$-th calling of $GSING(BV, CV)$. It is clear that if any one of $G_1, \hdots, G_{i-1}$ is singular, then $G$ is singular. Note that if any of $G_1, \hdots, G_{i-1}$ is empty then it is to be considered as nonsingular by convention. Now, $G_i$ is the induced subgraph (possibly with loops) whose vertex-set is selected as pendant vertex-set during $i$-th calling of $GSING(BV, CV)$. Let $i$ be the cut-vertex in the selected pendant vertex-set. Let $\alpha$ be the weight of loop at $i,$ if there is no loop at $i$ then $\alpha=0.$ $M^{\star}$ is the matrix corresponding to induced subgraph (with possible loops) of $G$, on the vertex-set $V(G)\setminus (V(G_1)\cup\hdots \cup V(G_{i}))$. Thus, if $G_1, \hdots, G_{i-1}$ are nonsingular, then $G$ is nonsingular if and only if the following submatrix is nonsingular,

\begin{equation}
\tilde{M}=\begin{bmatrix}
 A(G_{i} \setminus {i}) & j & O\\ j^T & \alpha & w^T \\ O^T & w & M^{\star}
\end{bmatrix}.
\end{equation}

 $A(G_i)$ is of form $J-D$, where $D=diag(d_1, d_2,\hdots,d_{|V(G_i)|})$ is diagonal matrix, $d_i=1-w(i), i=1,2,\hdots, |V(G_i)|.$ 

\begin{enumerate}

\item If more than one diagonal entries of $A(G_i \setminus i)$ are equal to 1, then as, at least two rows and two columns of $A(G_i \setminus i)$ are same, it is obvious that $G$ is singular. 

\item If exactly one diagonal entry of  $A(G_i \setminus i)$ is 1, the using  Theorem \ref{mtype}.1., $A(G_i \setminus i)$ is nonsingular. By Theorem \ref{mtype}.4., $\tilde{M}$ can be transformed to the following matrix,  
\begin{equation*}
\begin{bmatrix}
 A(G_i\setminus i) & o & O\\ o^T & \alpha-1 & w^T \\ O^T & w & M^{\star}
\end{bmatrix}.
\end{equation*}
Thus, $\tilde{M}$ is nonsingular if and only if the following nonsingular matrix is nonsingular,  $$\begin{bmatrix}
  \alpha-1 & w^T \\  w & M^{\star}
\end{bmatrix}.$$

\item If $w(i) \neq 1, i=1,\hdots, |V(G_i)-1|$ and $S=\sum_{i=1}^{|V(G_i)-1|}\frac{1}{1-w(i)}\neq 1$. Then $A(G_i \setminus i)$ is nonsingular. By  Theorem \ref{mtype}.4., $\tilde{M}$ can be transformed to the following matrix,  
\begin{equation*}
\begin{bmatrix}
 A(G_i\setminus i) & o & O\\ o^T & \alpha-\frac{S}{S-1} & w^T \\ O^T & w & M^{\star}
\end{bmatrix}.
\end{equation*}
Thus, $\tilde{M}$ is nonsingular if and only if the following nonsingular matrix is nonsingular,  $$\begin{bmatrix}
  \alpha-\frac{S}{S-1} & w^T \\  w & M^{\star}
\end{bmatrix}.$$

\item  If $A(G_i\setminus i)$ is singular. Then by Theorem \ref{mtype}.3.(b), $A(G_i)$ is nonsingular. The graph corresponding to $\tilde{M}$ is a coalescence of $G_i$ with the induced graph (possibly with loop) of $M^{\star}$.

 As, $A(G_i\setminus i)$ is singular and $A(G_i)$ is nonsingular, by Lemma \ref{supfirst},  $\tilde{M}$ is nonsingular if and only if $M^{\star}$ is nonsingular.
\end{enumerate}

This completes the proof.\\

\textbf{Acknowledgments}\\
The authors are grateful to Dr. Cheng Zheng for his valuable comments and suggestions.

\bibliographystyle{splncs03}
\bibliography{RN} 

\end{document}